\documentclass[preprint,11pt]{elsarticle}

\usepackage{amssymb}

\usepackage{threeparttable}
\usepackage{amssymb,amsmath,amsthm}
\usepackage{latexsym,graphicx,hyperref,epstopdf,multirow,hhline,xcolor,enumitem,algorithm,algpseudocode,float}
\usepackage[font=small,skip=-2pt]{caption}
\usepackage{changepage}
\usepackage{mathrsfs}
\usepackage{subfigure}
\usepackage{booktabs}
\usepackage{url}
\usepackage{setspace}

\usepackage{tikz}
\usepackage{pgfgantt}
\usepackage{adjustbox}
\usetikzlibrary{shapes,arrows,fit,calc,positioning}
\usetikzlibrary{decorations.pathreplacing}
\usepackage{setspace}

\usepackage[bindingoffset=0.2in,%
            left=1in,right=1in,top=1in,bottom=1in,%
            footskip=.25in]{geometry}

\hypersetup{pdfauthor={Name}}

\newtheorem{theorem}{Theorem}
\newtheorem{lemma}{Lemma}

\allowdisplaybreaks
\algnewcommand{\LineComment}[1]{\Statex \(\triangleright\) #1}

\journal{TBD}

\begin{document}

\begin{frontmatter}



\title{An efficient polynomial-time approximation scheme for parallel multi-stage open shops}


\author[a]{Jianming Dong}
\ead{djm226@163.com}
\author[b]{Ruyan Jin}
\ead{jry623@163.com}

\author[c]{Guohui Lin}
\ead{guohui@alberta.ca}
\author[d]{Bing Su}
\ead{subing684@sohu.com}
\author[e]{Weitian Tong\corref{cor1}}
\ead{wtong.research@gmail.com}
\author[e]{Yao Xu}
\ead{yxu@georgiasouthern.edu}

\cortext[cor1]{Corresponding Author}


\address[a]{Department of Mathematics, Zhejiang Sci-Tech University, Hangzhou, Zhejiang, China, 310018}
\address[b]{Public course teaching and Research  Office, Taizhou Vocational and Technical College, Taizhou, Zhejiang, China, 318000}
\address[c]{Department of Computing Science, University of Alberta, Edmonton, Alberta, Canada, T6G 2E8}
\address[d]{School of Economics and Management, Xi'an Technological University, Xi’an, China}
\address[e]{Department of Computer Science, Georgia Southern University, Statesboro, GA, USA, 30458}

\begin{abstract}
Various new scheduling problems have been arising from practical production processes 
and spawning new research areas in the scheduling field.
We study the parallel multi-stage open shops problem,
which generalizes the classic open shop scheduling and parallel machine scheduling problems. 
Given $m$ identical $k$-stage open shops and a set of $n$ jobs, 
we aim to process all jobs on these open shops 
with the minimum makespan, i.e., the completion time of the last job,
under the constraint
that job preemption is not allowed.
We present an efficient polynomial-time approximation scheme (EPTAS) 
for the case when both $m$ and $k$ are constant.
The main idea for our EPTAS is the combination of 
several categorization, scaling, and linear programming rounding techniques. 
Jobs and/or operations are first scaled and then categorized carefully into multiple types so that 
different types of jobs and/or operations are scheduled appropriately 
without increasing the makespan too much. 
\end{abstract}



\begin{keyword}



Scheduling; Parallel multi-stage open shops; Makespan; Linear programming; Efficient polynomial-time approximation scheme

\end{keyword}

\end{frontmatter}



\section{Introduction}\label{sec_introduction}

Scheduling plays a crucial role in manufacturing and service industries.
Incorporating new features that reflect modern industrial processes to classic models 
has been receiving considerable attentions from researchers. 
Open shop scheduling \cite{GS76} is one of the main branches of scheduling problems.  
An open shop consists of multiple machines (or processors), 
each of which processes a different operation (or task) of every job,
and every job has to go though all machines without ordering restrictions on which operation should be processed first.  
Considering the situation where multiple identical open shops are serving in production processes, 
we investigate \emph{$m$-parallel $k$-stage open shops} problem.

Formally, we are given $m$ identical $k$-stage open shops $\{S_1, S_2, \ldots, S_m \}$ 
and a set of $n$ jobs $\mathcal{J}=\{J_1, J_2, \ldots, J_n\}$.
Each $k$-stage open shop consists of $k$ machines.
Denote the open shop $S_\ell$ by its $k$ machines, i.e., $S_\ell = \{M_{\ell, 1}, M_{\ell, 2}, \ldots, M_{\ell, k} \}$.
Every job $J_i$ contains $k$ operations $\{O_{i, 1}, O_{i, 2}, \ldots, O_{i, k}\}$
and the operation $O_{i, j}$ needs to be processed by the machine $M_{\ell, j}$
if $J_i$ is assigned to the open shop $S_\ell$.
Without loss of generality, assume $O_{i, j}$ requires to occupy a $j$-th stage machine $p_{i, j}$ time non-preemptively, 
$\forall ~i\in [n], j \in [k]$.
Here $[x]$ denotes the number set $\{1, 2, \ldots, x\}$ for any positive integer $x$.
Then we represent each job $J_i$ as a $k$-tuple $(p_{i,1}, p_{i,2}, \ldots, p_{i,k})$.
Once a job is assigned to an open shop, it has to stay on this open shop until the job is completed. 
The goal is to minimize the \emph{makespan}, i.e., the completion time of the last job. 
Following the standard $3$-field notation $\alpha \mid \beta \mid \gamma$ introduced by Graham et al. \cite{GLL79},
where 
$\alpha$ describes the machine environment, 
$\beta$ refers to job characteristics, 
and $\gamma$ indicates objectives or optimality criteria,
the \emph{$m$-parallel $k$-stage open shops} problem is represented as $P_m(O_k) \mid \mid C_{\max}$
when both $m$ and $k$ are constant.

%
For a minimization optimization problem $\Pi$, 
a polynomial-time algorithm ${\cal A}$ has an approximation ratio $\alpha$ where
$\alpha = \min_{I \in \mathcal{I}} \{ {\cal A}(I)/{\mbox{OPT}(I)} \}$, 
where $\mathcal{I}$ is the set of all instances. 
If there is a family of algorithms $\{ \mathcal{A}_{\epsilon}, \epsilon > 0\}$ such that 
each algorithm $\mathcal{A}_{\epsilon}$ has an approximation ratio of $1 + \epsilon$ for any $\epsilon > 0$
and its time complexity is $O(n^{f(1/\epsilon)})$ for the instance size $n$ and some polynomial-time computable function $f$,
then the minimization problem $\Pi$ admits a \emph{polynomial-time approximation scheme} (PTAS).
We say a PTAS is an \emph{efficient polynomial-time approximation scheme} (EPTAS) if 
the running time of $\mathcal{A}_{\epsilon}$ is in the form of $f(1/\epsilon)\cdot n^{O(1)}$.
A \emph{fully polynomial-time approximation scheme} (FPTAS) is a PTAS with
its running time in form of $(1/\epsilon)^{O(1)} \cdot n^{O(1)}$,
i.e., a polynomial in $\frac{1}{\epsilon}$ and $n$.

Previously, only the case when $k=2$ was investigated for $P_m(O_k) \mid \mid C_{\max}$.
Chen et al. \cite{CZC13} introduced the parallel open shops problem $P_m(O_2) \mid \mid C_{\max}$ for the first time.
Inspired by the observation found by Gonzalez and Sahni \cite{GS76}
that the optimal makespan of $O_2 \mid \mid C_{\max}$ is either determined by a single job 
or determined by the total running time on one machine of the open shop,
Chen et al. \cite{CZC13} proposed a $3/2$-approximation algorithm for $P_2(O_2) \mid \mid C_{\max}$
and a $2$-approximation algorithm for $P_m(O_2) \mid \mid C_{\max}$ for any $m$ 
even when $m$ is part of the input.
Both these two approximation algorithms run in $O(n \log n)$ time. 
Later, Dong et al. \cite{DJH19} presented a pseudo-polynomial time dynamic programming algorithm
and adopted a standard scaling technique to develop an FPTAS for $P_m(O_2) \mid \mid C_{\max}$ 
when the number of open shops is any fixed constant.
Dong et al.'s algorithm was inspired 
by Gonzalez and Sahni \cite{GS76}'s linear-time optimal algorithm for $O_2 \mid \mid C_{\max}$,
the two-stage open flowshop scheduling with makespan minimization.
Because $O_2 \mid \mid C_{\max}$ can be solved optimally, 
Dong et al. \cite{DJH19} defined an important concept so-called \emph{critical-job}
to help determining the makespan on one two-stage open shop quantitatively.
When $k\geq 3$, 
the multi-stage open flowshop scheduling with makespan minimization (or $O_k \mid \mid C_{\max}$) 
becomes (weakly) NP-hard \cite{GS76},
which makes the problem $O_k \mid \mid C_{\max}$ loss good properties related to \emph{critical-job}.
Thus, Dong et al.'s idea for FPTAS \cite{DJH19} cannot be extended to the general $P_m(O_k) \mid \mid C_{\max}$
even when both $m$ and $k$ are constants.

Our main contribution is an EPTAS for $P_m(O_k) \mid \mid C_{\max}$. 
We combine several categorization, scaling, and linear programming rounding techniques,
some of which are inspired by the PTAS for the open shop problem $O_k \mid \mid C_{\max}$ \cite{SW98}
and the PTAS for the flow shop problem $F_k \mid \mid C_{\max}$  \cite{Hal98}.
First, the processing time of each job is scaled in a way 
such that we only need to consider the schedules with makespan at most $1$.
Then the job set is categorized into a set of big jobs and a set of small jobs.
This categorization makes sure the number of big jobs is a constant,
which enables to enumerate all possible schedules of big jobs in constant time.
Next, the operations of small jobs are categorized carefully into different types, 
which will be scheduled in different phases. 
Roughly, the operation of a small job is 
either simply scheduled to the end of the current schedule 
or ``densely'' fitted into the ``gaps'' or idle time intervals introduced by the schedule of big jobs.
The above scaling and categorization steps aim to define an abstract representative for each feasible schedule.
Instead of inspecting an exponential number of all feasible schedules, 
our EPTAS only checks a polynomial number of abstract representatives
and converts each abstract representative into a near-optimal feasible schedule within the represented group of schedules.
To obtain the smallest possible makespan without increasing the time complexity too much, 
each step of our EPTAS is well-designed and 
the overall analysis is involved. 
More details can be found in Section \ref{sec_eptas}.

The remaining context is organized as follows. 
In Section \ref{sec_relatedwork}, we review most related works. 
Section \ref{sec_preliminary} introduces necessary terminologies, concepts, and preprocessing steps.
The details of our EPTAS and its analysis are provided in Section \ref{sec_eptas}.
Finally, we make a conclusion in Section \ref{sec_conclusion}.

\section{Related Work}\label{sec_relatedwork}
%
It is easy to observe that $P_m(O_k) \mid \mid C_{\max}$ reduces 
to the $k$-stage open shop scheduling problem (denoted as $O_k \mid\mid C_{\max}$)  \cite{GS76} when $m=1$
and to the $m$-parallel identical machine scheduling problem (denoted as $P_m \mid\mid C_{\max}$) \cite{GJ79} when $k = 1$.
Both have been explored extensively in the literature.

For $O_k \mid\mid C_{\max}$, it admits a linear-time optimal algorithm when $k=2$
but becomes weakly NP-hard when $k\geq 3$ \cite{GS76}.
It is still an open question whether $O_k \mid\mid C_{\max}$ with $k \geq 3$ is strongly NP-hard \cite{Woe18}. 
Sevastianov and Woeginger \cite{SW98} presented a PTAS for the case
when $k \geq 3$ is a constant.
This problem is denoted as $O \mid\mid C_{\max}$ when $k$ is part of the input.
Williamson et al. \cite{WHH97} showed its strong NP-hardness and an inapproximability of $1.25$.
For other results on the open shop problem, we refer the readers to the survey by Ellur and Ramasamy \cite{AP15}.

For $P_m \mid\mid C_{\max}$, 
Graham \cite{Gra66} proposed the famous {\sc List-Scheduling} algorithm with an approximation ratio $2-1/m$,
for arbitrary $m$.
When $m\geq 2$ is a constant, Garey and Johnson \cite{GJ79} 
proved $P_m \mid\mid C_{\max}$ is weakly NP-hard and
designed a pseudo-polynomial time exact algorithm.
Sahni \cite{Sah76} presented an FPTAS. 
Hochbaum and Shmoys \cite{HS87} studied the case when $m$ is part of the input, denoted as $P \mid\mid C_{\max}$.
They proved its strong NP-hardness and proposed a PTAS.
Please refer to the survey \cite{Mok01} for more results.

Chen and Strusevich \cite{CS93} introduced 
the \emph{multiprocessor open shop} problem,
which generalizes $O_k \mid\mid C_{\max}$ and $P_m \mid\mid C_{\max}$ in another way.
Instead of allowing multiple identical open shops in $P_m(O_k) \mid\mid C_{\max}$, 
the multiprocessor open shop problem considers only one open shop and 
allows each stage to have multiple identical machines. 
The multiprocessor open shop problem is also studied under the names 
\emph{flexible open shop}, \emph{multi-machine open shop}, 
or \emph{open shop with parallel machines}.
Under the standard $3$-field notation, the multiprocessor open shop problem can be denoted as $O_k(P) \mid\mid C_{\max}$.
Chen and Strusevich \cite{CS93} designed a $(2-2/m^2)$-approximation algorithm for $O_2(P) \mid\mid C_{\max}$,
where $m$ is the upper bound of the number of machines in each stage. 
For the case when $k$ is part of input, Chen and Strusevich \cite{CS93} presented a $(2+\epsilon)$-approximation algorithm. 
These two approximation ratios were improved to $3/2+\epsilon$ and $2$ respectively by Woeginger \cite{SW99}.
When $k$ is a constant and the number of machines in each stage is part of input, 
Jansen and Sviridenko \cite{JS00} designed a PTAS.
For more results on the multiprocessor open shop problem and its variants, 
we would like to refer readers to a recent survey by Adak et al. \cite{AAB20}.

Changing the open shop environment to flow shop environment in $P_m(O_k) \mid\mid C_{\max}$ and $O_k(P) \mid\mid C_{\max}$ 
results in two more variants, 
i.e. the \emph{$m$-parallel $k$-stage flow shops} problem and the \emph{multiprocessor flow shop} problem 
denoted as $P_m(F_k) \mid\mid C_{\max}$ and $F_k(P) \mid\mid C_{\max}$ respectively under the $3$-field notation.
It has attracted quite a few attentions \cite{Kov85,HKA96,VE00,CC03,Als04,ZV12,DTL17,DJL20,TMG18,WCW19a,WCW19b,WCW20a,WCW20b}. 
For the case $k=2$, 
FPTASes were designed independently by Kovalyov \cite{Kov85}, Dong et al. \cite{DTL17}, and Wu et al. \cite{WCW19a} 
when $m$ is a fixed constant;
and Dong et al. \cite{DJL20} presented a PTAS when $m$ is part of the input.
For the case $k\geq 3$, 
Tong et al. \cite{TMG18} proposed a PTAS when both $m$ and $k$ are constant.
All the mentioned FPTASes and PTASes are best possible approximability result
as the \emph{$m$-parallel $k$-stage flow shops} problem is weakly NP-hard when $m\geq 2$ 
and becomes strongly NP-hard if $m$ is part of the input or $k\geq 3$ \cite{DJL20}. 
$F_k(P) \mid\mid C_{\max}$ and its special cases also attracted considerable attentions from researchers,
including but not limited to \cite{Gup88,GT91,Che95,HLV96,GHP97,Hal98,SW00,Wan05,RV10}. 
In particular, Hall \cite{Hal98} claimed a PTAS for $F_k(P) \mid\mid C_{\max}$ 
when $k$ is constant and the number of machines in each stage is also constant.
Ruiz and V{\'a}zquez-Rodr{\'\i}guez \cite{RV10} surveyed plenty of heuristics for the general $F_k(P) \mid\mid C_{\max}$ problem.

\section{Preliminary}\label{sec_preliminary}

%
Recall that each operation $O_{i, j}$, $i \in [n], j \in [k]$, has a processing time $p_{i, j}$. 
Define $P_{i} = \sum_{j=1}^{k} p_{i,j}$ to be the total processing time of job $J_{i}$ on an open shop 
and $P = \sum_{i=1}^{n} P_{i}$ to be the total processing time of all jobs.
Suppose a $P_m(O_k) || C_{\max}$ instance is denoted by its job set $\mathcal{J}$.
Given any instance $\mathcal{J}$,
let $\pi(\mathcal{J})$ denote a (feasible) schedule 
and $C_{\max}^{\pi(\mathcal{J})}$ be its makespan. 
Similarly, we define $\pi^*(\mathcal{J})$
and $C_{\max}^{\pi^*(\mathcal{J})}$ as the optimum schedule and its makespan respectively.
$\mathcal{J}$ will be omitted if there is no confusion in the context.

Because every job needs to go through all machines of a $k$-stage open shop,
we have $C_{\max}^{\pi^*} \geq \max\limits_{1\leq i\leq n}{P_{i}}$.
As the average workload is $\frac{P}{mk}$,
$C_{\max}^{\pi^*} \geq \frac{P}{mk}$.
On the other hand, an instance of the $m$ parallel machine scheduling problem can be constructed 
by adding a constraint that every open shop must complete a job before starting processing another job.
Then the famous \textsc{List-Scheduling} algorithm \cite{Gra66}
produces a schedule with makespan at most 
$\frac{P}{m}+ (1-\frac{1}{m})\max\limits_{1\leq i\leq n}{P_{i}} < \frac{P}{m}+\max\limits_{1\leq i\leq n}{P_{i}}$.
To wrap up, $C_{\max}^{\pi^*}$ for parallel open shops can be bounded in Lemma \ref{lemma_opt_bounds},
\begin{lemma}\label{lemma_opt_bounds}
For the $P_m(O_k) || C_{\max}$ problem, we have the following bounds for the minimum makespan $C_{\max}^{\pi^*}$.
\begin{equation}
\label{eq1}
\max\left\{\displaystyle\frac{P}{mk}, \max\limits_{1\leq i\leq n}{P_{i}}\right\} 
\leq C_{\max}^{\pi^*}
\leq \displaystyle\frac{P}{m}+\max\limits_{1\leq i\leq n}{P_{i}}
\end{equation}
\end{lemma}

We scale the processing time of each operation by
dividing $2\cdot \max \left\{ \frac{P}{m}, \max\limits_{1\leq i\leq n}{P_{i}} \right\}$.
Then the optimal makespan $C_{\max}^{*}$ can be normalized as 
\begin{equation}\label{eq_scale}
\frac{1}{2k} \leq C_{\max}^{\pi^*} \leq 1.
\end{equation}

In the sequel, we consider only those feasible schedules having a makespan less than or equal to $1$.

Using a parameter $\gamma$, which will be determined later in Algorithm {\sc Reduce}, 
we categorize the job set $\mathcal{J}$ into two sets, 
named as \emph{big job} set $\mathcal{B}^{\gamma}$ and \emph{small job} set $\mathcal{S}^{\gamma}$, respectively. 
If $\gamma$ is clear from the context, we may leave out the superscript for simplicity. 
\begin{eqnarray*}
\mathcal{B}^{\gamma} & = & \{J_{i}\in \mathcal{J} \mid \exists ~j \in[k], ~p_{i,j} \geq \gamma \}, \\
\mathcal{S}^{\gamma} & = & \mathcal{J} - \mathcal{B}^{\gamma}.
\end{eqnarray*}

As there are $mk$ machines in total and each machine has a load at most $1$,
the total processing time of all jobs is at most $mk$.
Each big job contains at least one operation with processing time greater than $\gamma$.
Thus, the number of big jobs is at most $mk/\gamma$,
which is summarized in Lemma \ref{lemma_num_big_jobs}.

\begin{lemma}\label{lemma_num_big_jobs}
Given $\gamma$, there are at most $ mk/\gamma $ big jobs.
\end{lemma}
%
%

\section{Algorithm description and analysis}\label{sec_eptas}
%
%
In this section, an EPTAS for the $P_m(O_k) || C_{\max}$ problem is presented.
We name it as {\sc Reduce-Assign-Dense-Greedy}, 
because it mainly consists of four procedures:
a reduction procedure to massage the given instance into a more structured one;
an assignment procedure to generate a ``rough'' and possibly infeasible schedule (named as assignment) for most jobs, 
a dense scheduling procedure to turn the assignment into a feasible schedule for those jobs;
a greedy procedure to greedily process the unassigned jobs and/or operations.
Roughly, the main purpose of the first two procedures is to categorize the feasible schedules into groups
and extract an abstract representative for each group of schedules. 
The last two procedures recover a feasible schedule from the abstract representative
such that the makespan is near-optimal with respect to the represented group of schedules. 
Detailed description and analysis for each procedure is presented 
in Sections \ref{subsec_reduce}, \ref{subsec_dense}, \ref{subsec_assign}, and \ref{subsec_eptas}, respectively.
%

\subsection{Description and Analysis for {\sc Reduce}}\label{subsec_reduce} 

Recall that any operation of a small job can be processed within $\gamma$ time by definition.
We further categorize the operations of small jobs into two types.
More specifically, 
we call an operation \emph{$\gamma^2$-small} if its processing time is less than $\gamma^2$, 
or \emph{$\gamma^2$-big} otherwise. 

Algorithm {\sc Reduce} identifies a specific value of $\gamma$ to categorize the job set.
Then, it converts any $P_m(O_k) || C_{\max}$ instance $\mathcal{J}$ into another instance $\mathcal{J}'$
such that 
(1) every operation of small jobs in $\mathcal{J}'$ has a processing time less than than $\gamma^2$;
(2) the processing time of every operation of big jobs in $\mathcal{J}'$ is a multiple of $\gamma^2$;
(3) the optimal makespan is increased by a little amount.

Consider a sequence of candidates for $\gamma$ 
\begin{equation}
\gamma_x = \delta^{2^x}, ~ \delta = \frac{\epsilon}{14 mk^3}, ~x \in \mathbb{Z}^{\geq 0}.
\end{equation}
We will search for an $x$ and assign $\gamma_x$ to $\gamma$.
We choose the value of $\epsilon$ such that $1/\epsilon$ is an integer, 
which implies both $mk/\delta$ and $mk/{\gamma_x}$ are integers as well.
Given a $\gamma_x$, 
we first categorize jobs into $\mathcal{B}^{\gamma_x}$ and $\mathcal{S}^{\gamma_x}$
and then define \emph{$\gamma_x^2$-big} and \emph{$\gamma_x^2$-small} operations for jobs in $\mathcal{S}^{\gamma_x}$.
%
%

\begin{center}
\begin{minipage}{0.9\textwidth}
\begin{algorithm}[H]
\caption*{\sc Reduce}
\vspace{-10pt}
\begin{flushleft}
\textbf{Input:} a general $P_m(O_k) || C_{\max}$ instance $\mathcal{J}$, parameter $\epsilon$;\\
\textbf{Output:} parameter $\gamma$, $\gamma^2$-big operation set $\mathcal{O}$,
        a special $P_m(O_k) || C_{\max}$ instance $\mathcal{J}'$ with 
        (1) $p_{i,j} \leq \gamma^2$ for $\forall j \in [k]$, $J_i \in \mathcal{S}^{\gamma}$;
        (2) $p_{i,j}$ is a multiple of $\gamma^2$ for $\forall j \in [k]$, $J_i \in \mathcal{B}^{\gamma}$.
\end{flushleft}
\vspace{-20pt}
\hrulefill
\begin{algorithmic}[1]
    \State Set $\delta = \frac{\epsilon}{14 mk^3}$;  
    \State Define a sequence of real numbers 
            $\gamma_x = \delta^{2^x}, x \in \mathbb{Z}^{\geq 0}$;
    \For{$x \in \{0, 1, \ldots, \frac{mk}{\delta}\}$} 
        \State Categorize $\mathcal{J}$ into $\mathcal{B}^{\gamma_x}$ and $\mathcal{S}^{\gamma_x}$;
        \State Let $\mathcal{O}$ denote the set of all $\gamma^2$-big operations;
        \State Let $L(\gamma_x)$ denote the total processing time of all operations in $\mathcal{O}_x$;
        \If{$L(\gamma_x) \leq \delta$}
            \State $\gamma \leftarrow \gamma_x$ 
            \For{$O_{i,j} \in \mathcal{O}$}
                \State $p_{i,j} \leftarrow 0$;
            \EndFor 
            \For{$J_i \in \mathcal{B}^{\gamma}$ and $O_{i,j} \in J_i$}
                \State $p_{i,j} \leftarrow \lceil p_{i,j} / \gamma^2 \rceil \cdot \gamma^2$; 
                \Comment Round $p_{i,j}$ up to the nearest multiple of $\gamma^2$
            \EndFor
            \State $\mathcal{J}' \leftarrow \mathcal{B}^{\gamma} \cup \mathcal{S}^{\gamma}$; 
            \State \Return $\gamma$, $\mathcal{O}$, and the new job set $\mathcal{J}'$. 
        \EndIf
    \EndFor
\end{algorithmic}
\end{algorithm}
\end{minipage}
\end{center}

\begin{lemma}\label{lemma_gamma}
Each operation in any job is $\gamma_x^2$-big for at most one index $x$.
\end{lemma}
\begin{proof}
Consider any job $J_i$ and follow the introduced rules for job categorization and operation categorization,
if there is any operation of $J_i$ is $\gamma_x^2$-big for an index $x$,
$J_i$ must be categorized as a small job by $\gamma_x$.
That is, 
$\forall ~j \in [k]$, $p_{i, j} < \gamma_x$ and 
$\exists ~j' \in [k]$, $p_{i, j'} \geq \gamma_x^2$. 
We argue that any operation of $J_i$ cannot be $\gamma_y^2$-big for any $y \not= x$.
Recall that $\gamma_x = \delta^{2^x}$ and $\delta = \frac{\epsilon}{14 mk^3} < 1$.

For any index $y > x$,  
the job $J_i$ is categorized as a big job with respect to $\gamma_{y}$ 
as  $\exists ~j' \in [k]$, $p_{i, j'} \geq \gamma_x^2 = \gamma_{x+1} \geq \gamma_{y}$.
Therefore, $J_i$ has no $\gamma_y^2$-big operations.

For any index $y < x$, 
the job $J_i$ is categorized as a small job with respect to $\gamma_{y}$ but 
every operation of $J_i$ has a processing time less than $\gamma_x = \gamma_{x-1}^2 \leq \gamma_{y}^2$,
which implies $J_i$ has no $\gamma_{y}^2$-big operations.

This completes the proof.
\end{proof}

\begin{lemma}\label{lemma_reduce}
Algorithm {\sc Reduce} terminates in $O(n/\epsilon)$  time
and generates a special instance $\mathcal{J}'$ with its optimal makespan bounded by
$$C_{\max}^{\pi^*(\mathcal{J}')} \leq C_{\max}^{\pi^*(\mathcal{J})} + mk^2\gamma.$$
\end{lemma}

\begin{proof}

We first analyze the time complexity of {\sc Reduce}.
Let $L(\gamma_x)$ denote the total processing time of all $\gamma_x^2$-big operations. 
Assume $L(\gamma_x) > \delta$ holds for all $x \in \{0, 1, \ldots, \frac{mk}{\delta}\}$.
By Lemma \ref{lemma_gamma},
the total processing time of all operations will be greater than
$\sum_{x=0}^{\frac{mk}{\delta}} \delta > mk$,
which contradicts to the fact that 
the total processing time of all operations is at most $mk$.
Therefore, there exists one $x < \frac{mk}{\delta}$ such that $L(\gamma_x) \leq \delta$.
Considering $\delta = \frac{\epsilon}{14 mk^3}$,
it takes $O(1/\epsilon)$ iterations to find such a special index $x$.
Since each iteration of the first for-loop takes $O(n)$ to compute $L(\cdot)$,
the first for-loop takes $O(n/\epsilon)$ time.
It is easy to observe that the second and third for-loops take $O(n)$ time in total.
Therefore, the time complexity of {\sc Reduce} is $O(n/\epsilon)$.

Now we show the upper bound of the optimal makespan of the constructed instance $\mathcal{J}'$.
Turning the processing time of each $\gamma^2$-big operation to zero 
does not increase the makespan.
Consider an optimal schedule $\pi^*(\mathcal{J})$ 
and assume jobs are scheduled according to $\pi^*(\mathcal{J})$.
Starting from time $0$, 
we scan through the operations 
on each open shop.
For every operation of big jobs, 
delaying the remaining schedule by less than $\gamma^2$ will generate 
enough space to process the rounded operation. 
By Lemma \ref{lemma_num_big_jobs}, there are $mk/\gamma$ big jobs and 
we delay the (partial) schedule on one open shop at most $mk^2 / \gamma$ times,
one delay for each operation of big jobs.
This results a feasible schedule for $\mathcal{J}'$ and increases the makespan by at most $mk^2 \gamma$.

This completes the proof.
\end{proof}

\subsection{Description and Analysis for {\sc Dense}}\label{subsec_dense}

In this section,
we introduce several essential concepts,
including \emph{restricted schedule} and \emph{assignment} of jobs. 
The main idea is to 
enumerate all possible restricted schedules for big jobs
and then use the ``gaps'', formed on each machine
between two consecutively scheduled operations,
to fit small jobs ``densely''.
%

\subsubsection{Restricted Schedule and Assignment}

Define a right half-open interval of length $\gamma^2$ as a $\gamma^2$-interval.
The time interval $[0,2)$ can be partitioned into $2/\gamma^2$ consecutive $\gamma^2$-intervals 
$[(t-1)\cdot \gamma^2, t\cdot \gamma^2), t \in [2/\gamma^2]$.
We say a feasible schedule is \emph{restricted} if 
every operation of big operations starts processing at the beginning of some $\gamma^2$-interval.

\begin{lemma}\label{lemma_restricted_schedule}
For any special $P_m(O_k) || C_{\max}$ instance $\mathcal{J}'$, 
the minimum makespan of restricted schedules is bounded above by 
$$C_{\max}^{\pi^*(\mathcal{J}')} + mk^2 \gamma.$$
\end{lemma}

\begin{proof}
The argument is similar to the proof for the upper bound in Lemma \ref{lemma_reduce}.
Starting from time $0$, 
we scan through the operations following the schedule $\pi^*(\mathcal{J}')$ on each open shop.
For every operation of big jobs, 
we delay the remaining schedule by less than $\gamma^2$ such that 
the operation starts processing at the beginning of some $\gamma^2$-interval.
As there are at most $mk^2 / \gamma$ times of delays,
the makespan is increased by at most  $mk^2 \gamma$.
\end{proof}

By Eq.(\ref{eq_scale}), Lemma \ref{lemma_reduce}, Lemma \ref{lemma_restricted_schedule},
and the definition of $\gamma$,
we consider special $P_m(O_k) || C_{\max}$ instances
and restricted schedules with makespan at most $2$ by default in the following context.

Given a feasible restricted schedule $\pi$, for each big job $J_i \in \mathcal{B}$, 
its \emph{assignment} is defined as $X_i = (\ell, \tau_1, \ldots, \tau_k)$,
where $\ell$ is the index of the open shop to which the job $J_i$ assigned in $\pi$,
and $\tau_j$ is the index of the $\gamma^2$-interval in which the $j$-th operation $O_{i,j}$ of $J_i$ starts processing. 
Then the multi-set of assignments for jobs in $\mathcal{B}$, denoted by $X_{\mathcal{B}} = \{X_i \mid J_i \in \mathcal{B}\}$,
is called an assignment for $\mathcal{B}$.
We say $X_{\mathcal{B}}$ is \emph{associated} with a feasible schedule $\pi$
if $X_{\mathcal{B}}$ can be defined from $\pi$.

\begin{lemma}\label{lemma_big_assignment}
There are at most $m^{mk/\gamma} \cdot (2/\gamma^2)^{mk^2/\gamma}$ distinct assignments for $\mathcal{B}$.
\end{lemma}

\begin{proof}
For each big job, there are at most $m (2/\gamma^2)^k $ distinct assignments.
From Lemma \ref{lemma_num_big_jobs}, the number of big jobs is at most $mk/\gamma$.
Therefore, the number of all possible assignments for $\mathcal{B}$ is no greater than 
$\left( m (2/\gamma^2)^k \right)^{mk/\gamma} = m^{mk/\gamma} \cdot (2/\gamma^2)^{mk^2/\gamma}$.
\end{proof}

We make a guess of $C_{\max}^{\pi^*(\mathcal{J}')}$, denoted by $C$,
which is used to define the assignment for small jobs.
Given an assignment for $\mathcal{B}$, 
we schedule the big jobs on open shops accordingly
with each operation starting processing at the beginning of some $\gamma^2$-interval.
Every machine, say $M_{\ell, j}$, may contain idle intervals, 
each formed between a pair of consecutively scheduled operations.
There may be two more intervals, one between the time 0 and the time starting processing the first operation on $M_{\ell, j}$,
the other one between completing the last operation and the guessed makespan $C$ on $M_{\ell, j}$.
Define such an idle interval as a \emph{gap}.
By Lemma \ref{lemma_num_big_jobs}, there are at most $mk/\gamma + 1$ gaps.
Let $G_{g, \ell, j} = [s_{g, \ell, j}, e_{g, \ell, j}]$ denote the $g$-th gap on the machine $M_{\ell, j}$,
where $s_{g, \ell, j}$ and $e_{g, \ell, j}$ are the starting and ending time of this gap respectively.

Given a feasible restricted schedule $\pi$, an associated assignment for $\mathcal{B}$, 
and an estimate $C$ of the makespan of $\pi$,
we define the \emph{assignment} for each small job $J_i \in \mathcal{S}$ as $X_i = (\ell, \tau_1, \ldots, \tau_k)$ similar to the assignment for a big job.
The difference is that
$\tau_j$ records the index of the gap instead of the $\gamma^2$-interval.
Let $X_{\mathcal{S}} = \{X_i \mid J_i \in \mathcal{S}\}$ denote the assignment for small jobs. 
Obviously, there may be an exponential number of distinct assignments for small jobs, as $|\mathcal{S}|$ may be linear in $n$.
We say $X_{\mathcal{S}}$ is \emph{associated} with $X_{\mathcal{B}}$ and $C$ 
if $X_{\mathcal{S}}$ can be defined from $X_{\mathcal{B}}$ and $C$.

$X_{\mathcal{B}}$ is said to be \emph{feasible} if it is associated with a feasible schedule.
$X_{\mathcal{S}}$ is said to be \emph{feasible} if 
every gap introduced by the associated $X_{\mathcal{B}}$ and $C$ provides enough space for
the operations of jobs $\mathcal{S}$ that are assigned to this gap,
i.e., 
\begin{equation}\label{eq_gap}
e_{g,\ell, j} - s_{g,\ell, j} \geq \sum_{J_i \in \mathcal{S}, X_i = (\ell, \tau_1, \ldots, \tau_j = g, \ldots, k)} p_{i,j}.
\end{equation}

As $C_{\max}^{\pi^*(\mathcal{J}')}$ is unknown, 
it is possible the guess $C$ is not correct. 
As long as $C\geq C_{\max}^{\pi^*(\mathcal{J}')}$, 
there must be a feasible assignment for small jobs.
When the guess $C$ is too small, say $C <  C_{\max}^{\pi^*(\mathcal{J}')}$, 
the defined $X_{\mathcal{S}}$ may be not feasible.
%
We will argue later in Algorithm {\sc Assign} to search 
 for an accurate enough guess such that 
 $C \leq C_{\max}^{\pi^*(\mathcal{J}')} + \epsilon$
 and a feasible assignment exists for most small jobs. 
%
%

\subsubsection{Algorithm {\sc Dense}}

Suppose $X_{\mathcal{B}}$ is associated with a feasible schedule $\pi$.
Assume $X_{\mathcal{S}}$ is also feasible and associated with $X_{\mathcal{B}}$ and $C$,
where $C$ is an estimate of the makespan of $\pi$.
Let $\mathcal{O}_{g, \ell, j}$ denote the set of operations assigned to start processing in the gap $G_{g, \ell, j}$,
$g\in [mk/\gamma + 1], \ell \in [m], j \in [k]$. 
Requiring the operations in $\mathcal{O}_{g, \ell, j}$ either scheduled in the gap $G_{g, \ell, j}$ or left unscheduled, 
we will generate a dense schedule for almost all operations in every gap.
In a dense schedule, 
any machine becomes idle only because there is no operation that is ready to be processed on the machine.
Note that the leftover operations from each set $\mathcal{O}_{g, \ell, j}$ 
will be processed later greedily in a post-processing stage. 

{\sc Dense} schedules operations while scanning through the gaps.
We mark a gap $G_{g, \ell, j}$ as ``scanned'' once
{\sc Dense} cannot find an operation in $\mathcal{O}_{g, \ell, j}$ to fit in the gap $G_{g, \ell, j}$.
Let $\mathcal{G}^{scanned}$ record the scanned gaps and
$\mathcal{O}^{remaining}$ collect the leftover operations from the set $\mathcal{O}_{g, \ell, j}$ 
after finishing scanning each gap $G_{g, \ell, j}$.
Both $\mathcal{G}^{scanned}$ and $\mathcal{O}^{remaining}$ are  initialized by empty sets.

{\sc Dense} (Lines 5-20) repeats the following steps until all gaps have been scanned. 
Every time some machine, say $M_{\ell, j}$, becomes idle at the earliest 
while ignoring the scanned gaps.
We update $T$ as this moment. 
Note that in the case when multiple machines becomes idle at the same time,
we pick on machine arbitrarily and $T$ may stay the same for multiple iterations.
Suppose $T \in [s_{g, \ell, j}, e_{g, \ell, j}]$ for some $g$ without loss of generality.
We check whether there exists an unprocessed operation $O_{i, j} \in \mathcal{O}_{g, \ell, j}$ 
such that no other operation of the corresponding job $J_i$ is currently being processed on some other machine of $S_\ell$. 
If at least one such operations exists, we pick one arbitrarily, start processing it at time $T$,
and remove it from $\mathcal{O}_{g, \ell, j}$.
Otherwise, $G_{g, \ell, j}$ and $\mathcal{O}_{g, \ell, j}$ are added to 
$\mathcal{G}^{scanned}$ and $\mathcal{O}^{remaining}$, respectively.
If $\mathcal{O}_{g, \ell, j}$ becomes empty, $G_{g, \ell, j}$ is also added to $\mathcal{G}^{scanned}$.

After all gaps has been scanned, 
the remaining operations in $\mathcal{O}^{remaining}$ will be processed
after the time $C$ sequentially
such that for each open shop only one machine is actively processing operations at any moment.
Refer to Algorithm {\sc Dense} for the detailed description. 
We will argue in Lemma \ref{lemma_dense} that for each gap most assigned operations will be scheduled in this gap.

\begin{center}
\begin{minipage}{0.9\textwidth}
\begin{algorithm}[H]
\caption*{\sc Dense} 
\vspace{-10pt}
\begin{flushleft}
\textbf{Input:} $C$,
                feasible $X_{\mathcal{B}}$,
                feasible $X_{\mathcal{S}}$;\\
\textbf{Output:}  
                a feasible schedule $\pi'$ with makespan at most $C + mk^2 \gamma + k \gamma^2$.
\end{flushleft}
\vspace{-20pt}
\hrulefill
\begin{algorithmic}[1]
    \State $\mathcal{G}^{scanned} = \emptyset$;
    \State $\mathcal{O}^{remaining} = \emptyset$;
    \State Let $\pi'$ denote the initial partial feasible schedule formed by $X_{\mathcal{B}}$;
    \State Let $T$ be the earliest idle time over all machines;
    \While{$T < C$}   

        \State Suppose $M_{\ell, j}$ is the earliest idle machine at $T$ without loss of generality;
        \State Suppose $T \in [s_{g, \ell, j}, e_{g, \ell, j}]$;
        \If{$\exists ~O_{i,j} \in \mathcal{O}_{g, \ell, j}$ such that 
            $p_{i,j} \leq e_{g, \ell, j} - T$ and \newline \phantom{~~~~~~~} 
            no other operations of $J_i$ is currently being processed on $S_\ell$}
            \State Start processing $O_{i,j}$ at time $T$;
            \State $\mathcal{O}_{g, \ell, j} \leftarrow \mathcal{O}_{g, \ell, j}  - O_{i,j}$; 
            \State Update $\pi'$ accordingly;  
        \Else \Comment Leave the remaining operations to the post-processing stage
            \State $\mathcal{G}^{scanned} \leftarrow \mathcal{G}^{scanned} + G_{g, \ell, j}$;
            \State $\mathcal{O}^{remaining} \leftarrow  \mathcal{O}^{remaining} + \mathcal{O}_{g, \ell, j}$; 
                \Comment $\mathcal{O}_{g, \ell, j}$ may be empty
        \EndIf
        
        \State Update $T$ as the earliest idle time over all machines 
                while skipping all scanned gaps;
       
    \EndWhile

    \State $T = C$; \Comment Start the post-processing stage to schedule $\mathcal{O}^{remaining}$
    \While{$\mathcal{O}^{remaining} \not= \emptyset$}
        \State Pick arbitrarily $\mathcal{O}_{g, \ell, j} \in \mathcal{O}^{remaining}$;
        \State Start processing all operations $\mathcal{O}_{g, \ell, j}$ in any order at time $T$ on $M_{\ell, j}$;
        \State $T \leftarrow T + \sum_{O_{i,j} \in \mathcal{O}_{g, \ell, j}} p_{i,j}$;
        \State $\mathcal{O}^{remaining} \leftarrow  \mathcal{O}^{remaining} - \mathcal{O}_{g, \ell, j}$;
        \State Update $\pi'$ accordingly;
    \EndWhile

    \State \Return a feasible schedule $\pi'$.
\end{algorithmic}
\end{algorithm}
\end{minipage}
\end{center}

\begin{lemma}\label{lemma_dense_key}
For any gap, there are less than $k$ leftover operations after the first While-loop of {\sc Dense}. 
\end{lemma}

\begin{proof}
A gap, say  $G_{g, \ell, j}$, is added to $\mathcal{G}^{scanned}$
when {\sc Dense} cannot find available operations to fit in this gap 
and all the remaining operations $\mathcal{O}_{g, \ell, j}$, which may be empty, 
will be delayed to be scheduled at the post-processing stage.
Therefore, there is at most one idle interval inside each gap after the first While-loop of {\sc Dense}.

For any gap $G_{g, \ell, j}$, if it has at least one operations remained after {\sc Dense} scans through the time range $[0, C]$,
suppose the idle interval formed inside $[s_{g, \ell, j}, e_{g, \ell, j}]$ after the execution of {\sc Dense}
is $[T_{g, \ell, j}, e_{g, \ell, j}]$.
Because both $X_{\mathcal{B}}$ and $X_{\mathcal{S}}$ are feasible.
By Eq.(\ref{eq_gap}),
every gap has enough space to process all operations assigned to this gap.
Every leftover operation from $\mathcal{O}_{g, \ell, j}$ cannot start processing at the time point $T_{g, \ell, j}$
because
another operation $O_{i,j'}$ belonging to the same job $J_i$ 
    is being processed on another machine of the same open shop.

Assume there are at least $k$ remaining operations in $\mathcal{O}_{g, \ell, j}$ after the first While-loop of {\sc Dense}. 
Each of these operations belongs to a different small job. 
At the time point $T_{g, \ell, j}$, at most $k - 1$ other machines of the open shop $S_\ell$ are busy.
This implies there is at least one of the leftover operations can be scheduled at $T_{g, \ell, j}$,
which is a contradiction.

This proves the lemma.
\end{proof}

\begin{lemma}\label{lemma_dense}
Suppose $X_{\mathcal{B}}$ is associated with a feasible restricted schedule $\pi$
and $X_{\mathcal{S}}$ is feasible and associated with $X_{\mathcal{B}}$ and $C$.
Algorithm {\sc Dense} takes $O(1/\gamma^2 + n^2)$ time and 
returns another feasible schedule $\pi'$ (which may be the same as $\pi$) such that 
$$C_{\max}^{\pi'(\mathcal{B} \cup \mathcal{S})} \leq C + mk^3 \gamma + k^2 \gamma^2.$$ 
\end{lemma}

\begin{proof}

By Lemma \ref{lemma_dense_key}, for any gap, 
there are less than $k$ leftover operations after the first While-loop of {\sc Dense}.
{\sc Dense} simply processes these leftover operations in any order after time $C$ on $M_{\ell, j}$,
which increases the makespan by at most $(k-1) \gamma^2 < k \gamma^2$.
Considering 
each open shops contains $k$ machines and
each machine has at most $mk/\gamma + 1$ gaps, 
the makespan of $\pi'$ is upper bounded by 
$(mk/\gamma + 1) \cdot k \cdot k \gamma^2 = mk^3 \gamma + k^2 \gamma^2.$

Now we analyze the time complexity of {\sc Dense}.
The first While-loop scans through the gaps and 
tries to fit every gap densely while maintaining the feasibility of the current schedule. 
The first While-loop starts a new iteration after either one operation is assigned to some gap or a gap completes scanning.
As there are at most $n$ small jobs and at most $(mk/\gamma + 1)\cdot mk = m^2k^2/\gamma + mk$ gaps,
the first While-loop has at most $nk + m^2k^2/\gamma + mk$ iterations.
Since each iteration needs $O(1/\gamma + n)$ time, 
the first While-loop takes $O(1/\gamma^2 + n^2)$ time.
The second While-loop simply processes the unprocessed operations from previous gaps sequentially after the time $C$
such that for each open shop
at most one machine is working at any moment,
which obviously takes linear time $O(n)$ and maintains the feasibility of the current schedule.
Therefore, the returned restricted schedule $\pi'$ is feasible
and the overall time complexity is $O(1/\gamma^2 + n^2)$.

\end{proof}

\subsection{Description and Analysis for {\sc Assign}}\label{subsec_assign}

{\sc Dense} is based on the given feasible assignments for big jobs and small jobs.
Nevertheless, there may be an exponential number of distinct assignments for small jobs.
In this section, 
given only $X_{\mathcal{B}}$ associated with a feasible schedule $\pi$,
we first search for an estimate $C$ of the makespan of $\pi$ such that $C \leq C_{\max}^{\pi(\mathcal{J}')} + \epsilon/2$
and then obtain a feasible assignment for most small jobs
by solving a linear program. 
The detailed procedure is described in Algorithm {\sc Assign}.

\begin{center}
\begin{minipage}{0.9\textwidth}
\begin{algorithm}[H]
\caption*{\sc Assign} 
\vspace{-10pt}
\begin{flushleft}
\textbf{Input:} parameter $\epsilon$, $\mathcal{J}'$, $X_{\mathcal{B}}$,
                which is associated with a feasible schedule $\pi$;\\
\textbf{Output:}  
                \begin{itemize}[noitemsep, topsep=0pt]
                \item  an estimate $C$ of the makespan of $\pi$ such that $C \leq C_{\max}^{\pi(\mathcal{J}')} + \epsilon/2$;
                \item  integrally assigned small jobs, denoted by $\mathcal{S}^{one}$;
                \item  fractionally assigned small jobs, denoted by $\mathcal{S}^{fractional}$;
                \item  feasible assignment for $\mathcal{S}^{one}$ associated with $X_{\mathcal{B}}$ and $C$, 
                        denoted by $X_{\mathcal{S}^{one}}$.
                \end{itemize}
\end{flushleft}
\vspace{-20pt}
\hrulefill
\begin{algorithmic}[1]
    \State Binary search for a $C$ value in the range $(0, 2)$ 
                such that the constructed LP has a basic feasible solution 
                and $C \leq C_{\max}^{\pi(\mathcal{J}')} + \epsilon/2$;
    \State Let $y$ be a basic feasible solution to the constructed LP with the chosen $C$ value;
    \State $\mathcal{S}^{one} = \{J_i \in \mathcal{S} \mid \exists ~y_{i, X} = 1\}$;
    \State $X_{\mathcal{S}^{one}} = \{X \mid \forall ~J_i \in \mathcal{S}^{one}, \exists ~y_{i, X} = 1\}$;
    \State $\mathcal{S}^{fractional} = \mathcal{S} \backslash \mathcal{S}^{one}$;
    \State \Return $C$, $\mathcal{S}^{one}$, $\mathcal{S}^{fractional}$, $X_{\mathcal{S}^{one}}$.
\end{algorithmic}
\end{algorithm}
\end{minipage}
\end{center}

Define the binary variable $y_{i, X}$ to indicate whether 
the job $J_i \in \mathcal{S}$ has an assignment $X$ in the gaps formed by $X_{\mathcal{B}}$ and $C$.
We construct a linear program, denoted by LP. 
The first set of constraints in LP makes sure every small job is assigned
and the second set of constraints in LP guarantees 
there is enough space to process all operations assigned to each gap. 
In the constructed LP, the number of non-trivial constraints is only $|{\cal S}| +  m^2k^2/\gamma + k m$  
while the number of variables $(mk/\gamma + 1)^k m |{\cal S}|$ is considerably larger.
Then, any basic feasible solution to LP has at most $|{\cal S}| + m^2k^2/\gamma + k m$ positive values.

\begin{center}
\begin{minipage}{\textwidth}
(LP)
\[
\begin{array}{rcll}
\displaystyle\sum_{X} y_{i, X}  &=      &1,             & \forall J_i \in {\cal S};\\
\displaystyle\sum_{J_i \in {\cal S}, X = (\ell, \tau_1, \ldots, \tau_j = g, \ldots, s_k)} p_{ij} y_{i, X}
                                &\le    &e_{g, \ell, j} - s_{g, \ell, j},    & \forall (g, \ell, j) \in [mk/\gamma + 1] \times [m] \times [k];\\
\displaystyle y                 &\ge    &0.             & 
\end{array}
\]
\end{minipage}
\end{center}

Given a basic feasible solution to LP,
let $\mathcal{S}^{one}$ and $\mathcal{S}^{fractional}$ be the set of small jobs that are assigned integrally and fractionally, respectively. 
From the definition, we have 
$|\mathcal{S}^{one}| + |\mathcal{S}^{fractional}| = |{\cal S}|.$
For each fractionally assigned small job, it corresponds to at least two positive variables.
The number of variables with positive values is at least $|\mathcal{S}^{one}| + 2 |\mathcal{S}^{fractional}|$.
Therefore, $|{\cal S}| + m^2k^2/\gamma + k m \geq |\mathcal{S}^{one}| + 2 |\mathcal{S}^{fractional}| = |{\cal S}| + |\mathcal{S}^{fractional}|$,
which implies
\begin{equation*}
|\mathcal{S}^{fractional}| \leq m^2k^2/\gamma + k m.
\end{equation*}

\begin{lemma}\label{lemma_assign}
Suppose $\mathcal{B}$ is the big job set of $\mathcal{J}'$ and 
$X_{\mathcal{B}}$ is associated with a feasible schedule $\pi$.
Algorithm {\sc Assign} takes  $n^{O(1)}\cdot (1/\gamma)^{O(1)}$ time to return
\begin{itemize}[noitemsep, topsep=0pt]
\item an estimate $C$ of the makespan of $\pi$ such that $C \leq C_{\max}^{\pi(\mathcal{J}')} + \epsilon/2$;
\item integrally assigned small jobs $\mathcal{S}^{one}$ and its feasible assignment  $X_{\mathcal{S}^{one}}$ associated with $X_{\mathcal{B}}$ and $C$;
\item fractionally assigned small jobs $\mathcal{S}^{fractional}$ with $|\mathcal{B}^{fractional}| \leq m^2k^2/\gamma + k m$.
\end{itemize}
\end{lemma}

\begin{proof}
We only need to argue for the time complexity and how to estimate $C$ accurately,
as the other parts of this lemma follows immediately from the previous discussions.

Given $X_{\mathcal{B}}$ and a $C$ value, 
the constructed LP contains $|{\cal S}| +  m^2k^2/\gamma + k m$ non-trivial constraints and  
$(mk/\gamma + 1)^k m |{\cal S}|$ variables.
It takes polynomial time in $n$ and $1/\gamma$, denoted by $n^{O(1)}\cdot (1/\gamma)^{O(1)}$, to solve this LP.

Whenever LP has a basic feasible solution, the conclusions regarding $\mathcal{S}^{one}$ and $\mathcal{S}^{fractional}$ hold.
That is, most small jobs can be assigned integrally. 
Due to the second set of constraints of LP, the assignment for $\mathcal{S}^{one}$ is feasible. 
Whether LP has a feasible solution can be used to tell whether the input $C$ value 
is able to generate a feasible assignment for $\mathcal{S}^{one}$. 
On the other hand, when $C \geq C_{\max}^{\pi(\mathcal{J}')}$, 
a feasible solution to LP can be easily constructed from the feasible schedule $\pi$
and thus a feasible assignment for $\mathcal{S}^{one}$ can be guaranteed.
Therefore, binary searching a value for $C$ in the range $(0, 2)$ gives us a desired precision of $C$ after $\log 1/\epsilon$ guesses.
The overall time complexity of {\sc Assign} is $n^{O(1)}\cdot (1/\gamma)^{O(1)}$.
\end{proof}

\subsection{Description and Analysis for Our EPTAS}\label{subsec_eptas}

\begin{center}
\begin{minipage}{0.9\textwidth}
\begin{algorithm}[H] \label{ptas}
\caption*{\sc Reduce-Assign-Dense-Greedy}
\vspace{-10pt}
\begin{flushleft}
\textbf{Input:} a $P_m(O_k) || C_{\max}$ instance $\mathcal{J}$, parameter $\epsilon$;\\
\textbf{Output:} a feasible schedule $\pi$ with makespan at most $(1+\epsilon) \cdot C_{\max}^{\pi^*(\mathcal{J})}$.
\end{flushleft}
\vspace{-20pt}
\hrulefill
\begin{algorithmic}[1]
    \State $\gamma$, $\gamma^2$-big operations $\mathcal{O}$, a special instance $\mathcal{J}'$ 
                $\leftarrow \mbox{{\sc Reduce}($\mathcal{J}$, $\epsilon$)}$;
    \State $\pi$ is initialized by None;
    \State Set $C_{\max} = \infty$;
    \For{any assignment $X_{\mathcal{B}}$ for big jobs}
        \State $C$, $\mathcal{S}^{one}$, $\mathcal{S}^{fractional}$, $X_{\mathcal{S}^{one}}$
                $\leftarrow \mbox{\sc Assign}(\epsilon, X_\mathcal{B}, \mathcal{J}')$;
        \State a feasible restricted schedule $\pi'$ for $\mathcal{B} \cup \mathcal{S}^{one}$
                $\leftarrow \mbox{\sc Dense}(X_{\mathcal{B}}, X_{\mathcal{S}^{one}}, C)$;     
        \State  Divide $\mathcal{S}^{fractional}$ evenly into $m$ partitions;
        \State Schedule each partition at the end of one open shop ``sequentially'';
        \If{$C_{\max} > C_{\max}^{\pi'}$}
            \State $\pi \leftarrow \pi'$;
            \State $C_{\max} \leftarrow C_{\max}^{\pi}$;
        \EndIf
    \EndFor
    \State Schedule the $\gamma^2$-big operations $\mathcal{O}$ at the end of open shops ``sequentially'';
    \State Update $\pi$ accordingly
    \State \Return the final schedule $\pi$
\end{algorithmic}
\end{algorithm}
\end{minipage}
\end{center}

Now we are ready to provide a detailed description and analysis for our EPTAS.
{\sc Reduce-Dense-Assign-Greedy} takes in any $P_m(O_k) || C_{\max}$ instance $\mathcal{J}$ and any small number $\epsilon \in (0, 1)$.
First, {\sc Reduce} is invoked to find an appropriate value of $\gamma$, 
        categorize jobs  $\mathcal{J}$ into a set  $\mathcal{B}$  of big jobs and a set  $\mathcal{S}$ of small jobs,
        obtain a set $\mathcal{O}$ of $\gamma^2$-big operations from small jobs, 
        and construct a special instance $\mathcal{J}'$.

Then, enumerating all possible assignments of big jobs and 
we repeat the following procedure for each assignment $X_{\mathcal{B}}$ 
to find a feasible restricted schedule $\pi$ with the minimum makespan.
For each assignment $X_{\mathcal{B}}$, 
{\sc Assign} guesses an estimate $C$ of the makespan of an associated feasible schedule
and generates (integral) assignments for most small jobs via solving a linear program.
$\mathcal{S}^{one}$ and $\mathcal{S}^{fractional}$ denote
the integrally and fractionally assigned small jobs. 
As the (integral) assignments for jobs in $X_{\mathcal{B}} \cup \mathcal{S}^{one}$ are known, 
{\sc Dense} is called to return a feasible restricted schedule for $X_{\mathcal{B}} \cup \mathcal{S}^{one}$.
The fractionally assigned small jobs $\mathcal{S}^{fractional}$ is then partitioned evenly into $m$ subsets,
each of which is greedily scheduled to the end of one open shop in a sequential manner
such that the open shop is dedicated to process only one job at any moment.

Finally, using the previously obtained feasible restricted schedule $\pi$ for the special instance $\mathcal{J}'$,
we construct a feasible schedule for the original instance $\mathcal{J}$.
We simply schedule each $\gamma^2$-big operation with its original processing time to the end of the open shop, 
to which its associated job is assigned,
such that the open shop is dedicated to process only one operation at any moment.

\begin{theorem}\label{thm_ptas}
Algorithm {\sc Reduce-Assign-Dense-Greedy} is an EPTAS. 
\end{theorem}

\begin{proof}
We first analyze the time complexity for {\sc Reduce-Assign-Dense-Greedy}.
By Lemma \ref{lemma_reduce}, {\sc Reduce} takes $O(n/\epsilon)$ time.
Lemma \ref{lemma_big_assignment} states
there are at most $m^{mk/\gamma} \cdot (2/\gamma^2)^{mk^2/\gamma}$ distinct assignments for $\mathcal{B}$,
which bounds the number of iterations of the For-loop.
Each iteration calls both {\sc Assign} and {\sc Dense},
taking $n^{O(1)} (1/\gamma)^{O(1)}$ and $O(1/\gamma^2 + n^2)$ time respectively.
Greedily scheduling the jobs in $\mathcal{S}^{fractional}$ or the $\gamma^2$-big operations $\mathcal{O}$ takes constant time.
Therefore, the overall time complexity of {\sc Reduce-Assign-Dense-Greedy} 
is $n^{O(1)} (1/\gamma)^{O(1/\gamma)}$,
where the term $(1/\gamma)^{O(1/\gamma)}$ is a function of $\epsilon$ by the definition of $\gamma$.

When we greedily schedule the jobs in $\mathcal{S}^{fractional}$ or the $\gamma^2$-big operations $\mathcal{O}$ to the end of open shops,
only one machine is active on every open shop and therefore the feasibility of the constructed schedule is maintained. 
Then the feasibility of the returned schedule follows from Lemma \ref{lemma_dense} and Lemma \ref{lemma_assign}.
Next, we estimate the quality of the returned schedule or the makespan.

By Lemma \ref{lemma_reduce} and Lemma \ref{lemma_restricted_schedule}, 
the optimal restricted schedule for the constructed special instance $\mathcal{J}'$ has a makespan bounded by 
\begin{equation}\label{eq_reduce_restrict}
C_{\max}^{\pi^*(\mathcal{J}')} \leq C_{\max}^{\pi^*(\mathcal{J})} + 2mk^2\gamma.
\end{equation}
Enumerating all possible assignments for big jobs,
including the one associated with the optimal restricted schedule $\pi^*(\mathcal{J}')$,
{\sc Reduce-Assign-Dense-Greedy} chooses the feasible schedule with the minimum makespan. 
Without loss of generality, we assume $X_{\mathcal{B}}$ is associated with $\pi^*(\mathcal{J}')$ in the sequel. 
From Lemma \ref{lemma_dense} and Lemma \ref{lemma_assign},
the makespan of $\pi$ when scheduling only $\mathcal{B} \cup \mathcal{S}^{one}$ is at most  
\begin{equation}\label{eq_dense_assign}
C_{\max}^{\pi(\mathcal{B} \cup \mathcal{S}^{one})} \leq C + mk^3 \gamma + k^2 \gamma^2
\leq C_{\max}^{\pi^*(\mathcal{J}')} + \epsilon/2 + mk^3 \gamma + k^2 \gamma^2.
\end{equation}

By Lemma \ref{lemma_assign}, $|\mathcal{B}^{fractional}| \leq m^2k^2/\gamma + k m$
and therefore at most $mk^2/\gamma + k $ small jobs are assigned to each open shop.
We greedily schedule these small jobs at the end of the current schedule $\pi(\mathcal{B} \cup \mathcal{S}^{one})$ in a sequential manner 
such that at most one machine is active on every open shop at any moment.
As the total processing time of a small job is at most $k \gamma^2$,
the makespan increases by at most
$(mk^2/\gamma + k) \cdot k \gamma^2$.
That is,
\begin{equation}\label{eq_greedy1}
C_{\max}^{\pi(\mathcal{J}')} \leq C_{\max}^{\pi(\mathcal{B} \cup \mathcal{S}^{one})} + (mk^2/\gamma + k) \cdot k \gamma^2.
\end{equation}

To obtain a feasible schedule for the original instance $\mathcal{J}$, 
we schedule the $\gamma^2$-big operations $\mathcal{O}$ at the end of $\pi(\mathcal{J}')$ sequentially.
As the total processing time of operations in $\mathcal{O}$ is at most $\delta$,
the makespan increases by at most $\delta$.
That is, 
\begin{equation}\label{eq_greedy2}
C_{\max}^{\pi(\mathcal{J})} \leq C_{\max}^{\pi(\mathcal{J}')} + \delta.
\end{equation}

Combining Eq.(\ref{eq_reduce_restrict}), Eq.(\ref{eq_dense_assign}), Eq.(\ref{eq_greedy1}),
and Eq.(\ref{eq_greedy2}),
the makespan of the returned schedule can be estimated as follows.
\begin{align*}
C_{\max}^{\pi(\mathcal{J})} 
& \leq C_{\max}^{\pi^*(\mathcal{J})} + 2mk^2\gamma + \epsilon/2 + mk^3 \gamma + k^2 \gamma^2 + (mk^2/\gamma + k) \cdot k \gamma^2  + \delta \\
& \leq C_{\max}^{\pi^*(\mathcal{J})} + \epsilon/2 + 7 mk^3 \delta \\
& \leq C_{\max}^{\pi^*(\mathcal{J})} + \epsilon, 
\end{align*}
where the last two inequalities are because of $\gamma = \delta^{2^x} \leq \delta = \frac{\epsilon}{14 mk^3} < 1$.

This completes the proof.
\end{proof}

\section{Conclusion}\label{sec_conclusion}

We investigate the parallel multi-stage open shops (denoted by $P_m(O_k) \mid\mid C_{\max}$) problem,
a meaningful hybrid of the classic open shop scheduling and parallel machine scheduling,
which takes in multiple identical open shops and schedules jobs to achieve the minimum makespan. 
Note that shifting a job between open shops and preemptive processing an operation are not allowed. 
As $P_m(O_k) \mid\mid C_{\max}$ can be treated as 
the parallelized version of the classic open shop scheduling problem,
it inherits the computational complexity directly from the latter problem
and thus it is NP-hard in the weak sense when $k \geq 3$ is a constant.
When both $m$ and $k$ are constant, 
we design an efficient polynomial-time approximation scheme (EPTAS) for $P_m(O_k) \mid\mid C_{\max}$.

There are two interesting open questions that worth further exploration in the future. 
\begin{enumerate}
\item Because the classic open shop scheduling problem is strongly NP-hard when $k \geq 3$ 
        has been open since the 1970s \cite{Woe18}.
        As a generalized variant, whether $P_m(O_k) \mid\mid C_{\max}$ is strongly NP-hard is also unknown. 
        On the other hand, $P_m \mid\mid C_{\max}$ is only weakly NP-hard when $m$ is fixed \cite{GJ79}.
        It is worth investigating whether $P_m(O_k) \mid\mid C_{\max}$ is NP-hard in a strong sense
        for constant $m$ and $k \geq 3$.
        If it is not strongly NP-hard, can we design an FPTAS?
\item When $m$ is part of the input, $P(O_k) \mid\mid C_{\max}$ generalizes 
        the parallel machine scheduling $P \mid\mid C_{\max}$
        and thus inherits the strong NP-hardness. 
        It would be interesting to investigate whether a PTAS exists for $P(O_k) \mid\mid C_{\max}$.
        Perhaps we can approach the problem starting from the special case when $k=2$.
\end{enumerate}

\section*{Acknowledgments}

JD is supported by National Social Science Foundation of China (Grant No. 11971435).
GL is supported by the NSERC Canada.
BS is supported by National Social Science Foundation of China (Grant No. 20XGL023) and Science and Technology Department of Shaanxi Province (Grant No. 2020JQ-654).
WT and YX are supported by the Office of Research, Georgia Southern University.




\bibliographystyle{abbrv}
\bibliography{manuscript-arXiv}

\begin{thebibliography}{10}

\bibitem{AAB20}
Z.~Adak, M.~{\"O}. Ar{\i}o{\u{g}}lu~Akan, and S.~Bulkan.
\newblock Multiprocessor open shop problem: literature review and future
  directions.
\newblock {\em Journal of Combinatorial Optimization}, 40:547--569, 2020.

\bibitem{Als04}
A.~Al-Salem.
\newblock A heuristic to minimize makespan in proportional parallel flow shops.
\newblock {\em International Journal of Computing \& Information Sciences},
  2(2):98, 2004.

\bibitem{AP15}
E.~Anand and R.~Panneerselvam.
\newblock Literature review of open shop scheduling problems.
\newblock {\em Intelligent Information Management}, 7(01):33, 2015.

\bibitem{CC03}
D.~Cao and M.~Chen.
\newblock Parallel flowshop scheduling using tabu search.
\newblock {\em International journal of production research},
  41(13):3059--3073, 2003.

\bibitem{Che95}
B.~Chen.
\newblock Analysis of classes of heuristics for scheduling a two-stage flow
  shop with parallel machines at one stage.
\newblock {\em Journal of the Operational Research Society}, 46(2):234--244,
  1995.

\bibitem{CS93}
B.~Chen and V.~A. Strusevich.
\newblock Worst-case analysis of heuristics for open shops with parallel
  machines.
\newblock {\em European Journal of Operational Research}, 70(3):379--390, 1993.

\bibitem{CZC13}
Y.~Chen, A.~Zhang, G.~Chen, and J.~Dong.
\newblock Approximation algorithms for parallel open shop scheduling.
\newblock {\em Information Processing Letters}, 113(7):220--224, 2013.

\bibitem{DJH19}
J.~Dong, R.~Jin, J.~Hu, and G.~Lin.
\newblock A fully polynomial time approximation scheme for scheduling on
  parallel identical two-stage openshops.
\newblock {\em Journal of Combinatorial Optimization}, 37(2):668--684, 2019.

\bibitem{DJL20}
J.~Dong, R.~Jin, T.~Luo, and W.~Tong.
\newblock A polynomial-time approximation scheme for an arbitrary number of
  parallel two-stage flow-shops.
\newblock {\em European Journal of Operational Research}, 281(1):16--24, 2020.

\bibitem{DTL17}
J.~Dong, W.~Tong, T.~Luo, X.~Wang, J.~Hu, Y.~Xu, and G.~Lin.
\newblock An {FPTAS} for the parallel two-stage flowshop problem.
\newblock {\em Theoretical computer science}, 657:64--72, 2017.

\bibitem{GJ79}
M.~R. Garey and D.~S. Johnson.
\newblock {\em Computers and intractability}, volume 174.
\newblock freeman San Francisco, 1979.

\bibitem{GS76}
T.~Gonzalez and S.~Sahni.
\newblock Open shop scheduling to minimize finish time.
\newblock {\em Journal of the ACM}, 23:665--679, 1976.

\bibitem{Gra66}
R.~L. Graham.
\newblock Bounds for certain multiprocessing anomalies.
\newblock {\em Bell system technical journal}, 45(9):1563--1581, 1966.

\bibitem{GLL79}
R.~L. Graham, E.~L. Lawler, J.~K. Lenstra, and A.~R. Kan.
\newblock Optimization and approximation in deterministic sequencing and
  scheduling: a survey.
\newblock In {\em Annals of discrete mathematics}, volume~5, pages 287--326.
  1979.

\bibitem{Gup88}
J.~N. Gupta.
\newblock Two-stage, hybrid flowshop scheduling problem.
\newblock {\em Journal of the operational Research Society}, 39(4):359--364,
  1988.

\bibitem{GHP97}
J.~N. Gupta, A.~Hariri, and C.~N. Potts.
\newblock Scheduling a two-stage hybrid flow shop with parallel machines at the
  first stage.
\newblock {\em Annals of Operations Research}, 69:171--191, 1997.

\bibitem{GT91}
J.~N. Gupta and E.~A. Tunc.
\newblock Schedules for a two-stage hybrid flowshop with parallel machines at
  the second stage.
\newblock {\em The International Journal of Production Research},
  29(7):1489--1502, 1991.

\bibitem{Hal98}
L.~A. Hall.
\newblock Approximability of flow shop scheduling.
\newblock {\em Mathematical Programming}, 82(1):175--190, 1998.

\bibitem{HKA96}
D.~W. He, A.~Kusiak, and A.~Artiba.
\newblock A scheduling problem in glass manufacturing.
\newblock {\em IIE Transactions}, 28:129--139, 1996.

\bibitem{HS87}
D.~S. Hochbaum and D.~B. Shmoys.
\newblock Using dual approximation algorithms for scheduling problems
  theoretical and practical results.
\newblock {\em Journal of the ACM}, 34(1):144--162, 1987.

\bibitem{HLV96}
J.~A. Hoogeveen, J.~K. Lenstra, and B.~Veltman.
\newblock Preemptive scheduling in a two-stage multiprocessor flow shop is
  {NP}-hard.
\newblock {\em European Journal of Operational Research}, 89(1):172--175, 1996.

\bibitem{JS00}
K.~Jansen and M.~I. Sviridenko.
\newblock Polynomial time approximation schemes for the multiprocessor open and
  flow shop scheduling problem.
\newblock In {\em Annual Symposium on Theoretical Aspects of Computer Science},
  pages 455--465, 2000.

\bibitem{Kov85}
M.~Y. Kovalyov.
\newblock Efficient epsilon-approximation algorithm for minimizing the makespan
  in a parallel two-stage system.
\newblock {\em Vesti Academii navuk Belaruskai SSR. Seria
  phizika-matematichnikh navuk}, 1985.

\bibitem{Mok01}
E.~Mokotoff.
\newblock Parallel machine scheduling problems: A survey.
\newblock {\em Asia-Pacific Journal of Operational Research}, 18(2):193, 2001.

\bibitem{RV10}
R.~Ruiz and J.~A. V{\'a}zquez-Rodr{\'\i}guez.
\newblock The hybrid flow shop scheduling problem.
\newblock {\em European journal of operational research}, 205(1):1--18, 2010.

\bibitem{Sah76}
S.~K. Sahni.
\newblock Algorithms for scheduling independent tasks.
\newblock {\em Journal of the ACM}, 23(1):116--127, 1976.

\bibitem{SW99}
P.~Schuurman and G.~J. Woeginger.
\newblock Approximation algorithms for the multiprocessor open shop scheduling
  problem.
\newblock {\em Operations Research Letters}, 24(4):157--163, 1999.

\bibitem{SW00}
P.~Schuurman and G.~J. Woeginger.
\newblock A polynomial time approximation scheme for the two-stage
  multiprocessor flow shop problem.
\newblock {\em Theoretical Computer Science}, 237(1-2):105--122, 2000.

\bibitem{SW98}
S.~V. Sevastianov and G.~J. Woeginger.
\newblock Makespan minimization in open shops: A polynomial time approximation
  scheme.
\newblock {\em Mathematical Programming}, 82(1):191--198, 1998.

\bibitem{TMG18}
W.~Tong, E.~Miyano, R.~Goebel, and G.~Lin.
\newblock An approximation scheme for minimizing the makespan of the parallel
  identical multi-stage flow-shops.
\newblock {\em Theoretical Computer Science}, 734:24--31, 2018.

\bibitem{VE00}
G.~Vairaktarakis and M.~Elhafsi.
\newblock The use of flowlines to simplify routing complexity in two-stage
  flowshops.
\newblock {\em Iie Transactions}, 32(8):687--699, 2000.

\bibitem{Wan05}
H.~Wang.
\newblock Flexible flow shop scheduling: optimum, heuristics and artificial
  intelligence solutions.
\newblock {\em Expert Systems}, 22(2):78--85, 2005.

\bibitem{WHH97}
D.~P. Williamson, L.~A. Hall, J.~A. Hoogeveen, C.~A. Hurkens, J.~K. Lenstra,
  S.~V. Sevast'janov, and D.~B. Shmoys.
\newblock Short shop schedules.
\newblock {\em Operations Research}, 45(2):288--294, 1997.

\bibitem{Woe18}
G.~J. Woeginger.
\newblock The open shop scheduling problem.
\newblock In {\em 35th Symposium on Theoretical Aspects of Computer Science
  (STACS 2018)}, 2018.

\bibitem{WCW19b}
G.~Wu, J.~Chen, and J.~Wang.
\newblock On scheduling inclined jobs on multiple two-stage flowshops.
\newblock {\em Theoretical Computer Science}, 786:67--77, 2019.

\bibitem{WCW19a}
G.~Wu, J.~Chen, and J.~Wang.
\newblock Scheduling two-stage jobs on multiple flowshops.
\newblock {\em Theoretical Computer Science}, 776:117--124, 2019.

\bibitem{WCW20a}
G.~Wu, J.~Chen, and J.~Wang.
\newblock Improved approximation algorithms for two-stage flowshops scheduling
  problem.
\newblock {\em Theoretical Computer Science}, 806:509--515, 2020.

\bibitem{WCW20b}
G.~Wu, J.~Chen, and J.~Wang.
\newblock On scheduling multiple two-stage flowshops.
\newblock {\em Theoretical Computer Science}, 818:74--82, 2020.

\bibitem{ZV12}
X.~Zhang and S.~van~de Velde.
\newblock Approximation algorithms for the parallel flow shop problem.
\newblock {\em European Journal of Operational Research}, 216(3):544--552,
  2012.

\end{thebibliography}





\end{document}